\documentclass[12pt,reqno]{article}
\usepackage{amsmath,amssymb,amsthm}
\usepackage{comment}

\usepackage[letterpaper]{geometry}

\usepackage[shortlabels]{enumitem}
  \setlist[enumerate,1]{leftmargin=25pt}
  \setlist[itemize,1]{leftmargin=15pt}
  \setlist[description,1]{leftmargin=15pt}

\usepackage{graphicx}
\usepackage[font=small]{caption}
\usepackage{floatrow}
\usepackage{subcaption}

\usepackage{hyperref}
\usepackage{mathabx}
\usepackage{quoting}

\usepackage{setspace}
\usepackage{url}

\newenvironment{Blue}{\noindent\color{blue}}{}
\newcommand\blue[1]{\begin{Blue}#1\end{Blue}}

\newenvironment{Red}{\noindent\color{red}}{}

\newtheorem{lemma}{Lemma}

\theoremstyle{definition}

\newcommand\Ar{\medskip\noindent\textbf{A:\ }}
\newcommand\bra[1]{\ensuremath{\langle#1|}}
\newcommand\braket[2]{\ensuremath{\langle#1\,|\,#2\rangle}}
\newcommand\C{\ensuremath{\mathbb C}}
\newcommand\Dom[1]{\ensuremath{\text{Dom}(#1)}}
\renewcommand\H{\ensuremath{\mathcal H}}

\newcommand{\ket}[1]{\ensuremath{|#1\rangle}}
\newcommand{\ketbra}[2]{\ensuremath{|#1\rangle\langle#2|}}

\renewcommand\phi{\varphi}
\newcommand\qef{\hfill$\triangleleft$} 
\newcommand{\Qn}{\medskip\noindent\textbf{Q:\ }}
\newcommand{\R}{\ensuremath{\mathbb R}}
\newcommand\Sec[1]{\section{\large #1}}

\renewcommand{\v}{\ensuremath{\vec{v}}}
\newcommand\Val[1]{\ensuremath{\text{Val}\big(#1\big)}}
\newcommand\val[1]{\ensuremath{\text{Val}(#1)}}

\title{What are kets?}
\author{Yuri Gurevich\\
\normalsize Computer Science \& Engineering\\
\normalsize University of Michigan
\and
Andreas Blass\\
\normalsize Mathematics\\
\normalsize University of Michigan}
\begin{document}
\newpage
\date{}
\maketitle
\thispagestyle{empty}

\begin{abstract}
According to Dirac's bra-ket notation, in an inner-product space, the inner product \braket xy of vectors $x,y$ can be viewed as an application of the bra \bra x to the ket \ket y.
Here \bra x is the linear functional $\ket y \mapsto \braket xy$ and \ket y is the vector $y$.
But often --- though not always --- there are advantages in seeing
\ket y as the function $a \mapsto a\cdot y$ where $a$ ranges over the
scalars. For example, the outer product \ketbra yx becomes simply the
composition $\ket y \circ \bra x$.
It would be most convenient to view kets sometimes as vectors and sometimes as functions, depending on the context. This turns out to be possible.

While the bra-ket notation arose in quantum mechanics, this note presupposes no familiarity with quantum mechanics.
\end{abstract}

\Sec{The question}\label{s:qn}

\noindent\textbf{Q}\footnote{Quisani is a former student of the first author.}:
Gentlemen, I have a question for you.
But first I need to motivate it and explain where I am coming from.

The question is related to the so-called inner-product spaces which are vector spaces over the field \R\ of real numbers or the field \C\ of complex numbers, furnished with inner product \braket xy,  also known as scalar product.
Euclidean spaces and (more generally) Hilbert spaces are the most familiar examples of inner-product spaces.
I'll stick to the case of \C\ which is of greater interest to me.

Let \H\ be an inner-product space over \C, and let $x,y$ range over the vectors in \H.
A vector $x$ gives rise to the linear functional $y \mapsto \braket xy$ that maps any vector $y$ to the scalar \braket xy.
This linear functional is called a \emph{bra} and denoted \bra x in Dirac's bra-ket notation, introduced by Paul Dirac \cite{Dirac}.  The vector $y$ is called a \emph{ket} and denoted \ket y.
Thus the inner product \braket xy can be viewed as the application $\bra x\, \ket y = \bra x (\ket y)$ of the bra \bra x to the ket \ket y.

By the way a side question occurs to me.
If kets are vectors and bras are linear functionals, then \braket xy can just be \emph{defined} as the application $\bra x\, \ket y$ in any vector space $V$, rather than presumed to exist.
Does the inner product contribute anything?

\noindent\textbf{A}\footnote{The authors speaking one at a time}:
It does.
It provides a particular embedding $\ket x \mapsto \bra x$ of our vector space \H\ to the vector space of linear functionals
on \H, which is the dual of \H\ in the theory of vector spaces.
Notice that the axiom $\braket xx \ge 0$ makes good sense for inner product spaces but not if we have no connection between \ket x and \bra x.
If \H\ is a Hilbert space, in particular if \H\ is finite dimensional, this embedding is an isomorphism.

\Qn Thanks.
Let me proceed to my main question.
As I said, kets are vectors according to Dirac, and the point of view that kets are vectors is ubiquitous.
Here is a quote from my favorite textbook on quantum computing: ``The notation \ket{\cdot} is used to indicate that the object is a vector'' \cite[p.~62]{NC}.

But recently I watched a recorded lecture \cite{Werner} by Reinhard Werner, a professor of physics, who defined \ket y as the linear function.

$a\mapsto a\cdot y$ from scalars to vectors, so that $y$ is $\ket y(1)$.

\Ar Did Prof.\ Werner compare the definitions?

\Qn No, there was only one definition of kets in his lecture%
\footnote{We asked Prof. Werner about the source of his definition.
He did not remember it.
However, the definition matches the spirit of category theory.
There, instead of dealing with the elements of a structure, one deals with maps into (or out of) the structure.}.
But I think that the outer product \ketbra yx demonstrates an advantage of his approach.
In the traditional approach, \ketbra yx is defined to be a (linear)
function by fiat.
For example, Nielsen and Chuang \cite[p.~68]{NC} write:
\begin{quoting}[vskip=1pt,leftmargin=\parindent]
``Suppose \ket v is a vector in an inner product space $V$, and \ket w is a vector in an inner product space $W$.
Define \ketbra wv to be the linear operator from $V$ to $W$ whose action is defined by
$ \big(\ketbra wv\big)(\ket{v'}) \equiv \ket w\,\braket v{v'}
 = \braket v{v'}\, \ket w $.''
\end{quoting}
I do not know how consistent Prof.\ Werner is in using his definition
of kets, but for the purpose of this discussion let me introduce a
purely functional approach where a ket \ket y is always the function $a\mapsto a\cdot y$.
In this approach, \ketbra yx is simply the composition $\ket y \circ
\bra x$ of two linear functions and thus \ketbra yx is naturally a
linear function, the same function that Nielsen and Chuang define.

At last, I come to my question: Is a ket a vector or a function?
It cannot be both, can it?

\Ar Well, abuse of notation is common in mathematics, especially if the meaning is obvious from the context.

\Qn Mathematically, the purely functional approach is attractive.
Since composition of functions is associative, we can drop the parentheses in expressions like
\[ \ket u\, \bra v\, \ket w\, \bra x\, \ket y\, \bra z. \]
But the purely functional approach has some problems.
Consider the scalar product \braket xy for example.
The composition $\bra x \circ \ket y$ is a linear operator on \C, not a scalar.
We can define, by fiat, that the original \braket xy is, in the purely functional approach, $\left(\bra x \circ \ket y\right)(1)$.
It would be more natural, of course, to view \ket y as a vector in the context of scalar product.

I wonder whether one can take advantage of both approaches even if this leads to abuse of notation.
Maybe there is a resolution of this abuse so that the intended meaning is obvious from the context.

\Sec{Dirac terms}\label{s:terms}

\Ar It seems that there are such resolutions.
We need some analysis to understand what is going on.
This discussion may be more pedantic than usual.

To keep the notation simple, we restrict attention to kets and bras over the same inner-product space \H.
For the purpose of the analysis, we put forward the following
tentative convention.

\paragraph{Tentative Convention.}
Every occurrence of a ket is marked as a \emph{vector ket} or a \emph{function ket}.
If $y$ is a vector then the vector ket \ket y denotes the vector $y$ but the corresponding function ket \ket y denotes the function $a\mapsto a\cdot y$.
More generally, for any label $L$, the vector ket \ket L denotes some vector \v\ in \H, while the function ket \ket L denotes the function $a \mapsto a \cdot \v$, from \C\ to \H, for the same vector \v. \qef

\Qn You can avoid marking by declaring that, by default, kets are function kets; the corresponding vector ket is \ket L(1).

\Ar This is correct.
We stick to marking because of its symmetry.
Your proposal may give an impression of bias in favor of function kets.

You have mentioned inner products \braket xy and outer products \ketbra yx.
Let's consider more general products of alternating kets and bras.

\paragraph{Dirac terms: syntax.}
Kets and bras are \emph{Dirac characters}.
By a \emph{Dirac term\footnotemark} we mean a nonempty sequence of Dirac characters where kets and bras alternate;
the sequence is often furnished with parentheses.

\footnotetext{While logicians speak about terms, computer scientists speak about expressions. Here we use the logicians' vocabulary for a very utilitarian reason: ``term'' is shorter than ``expression.''}

More formally, Dirac terms are defined inductively.
Dirac characters are terms.
A concatenation $s_1s_2$ of Dirac terms $s_1, s_2$ where  kets and bras alternate is a Dirac term.

By default, a Dirac term $s$ comes with enough parentheses to parse $s$, i.e.\ to determine how $s$ is constructed  from kets and bras by means of concatenation.

\Qn Suppose some ket \ket L occurs more than once in a Dirac term.
Can some occurrences of \ket L be vector kets and some function kets?

\Ar Sure, why not?

\paragraph{Dirac terms: semantics.}\mbox{}

\Qn Since kets are disambiguated in the Tentative Convention, semantics seems obvious.

\Ar It is obvious, but we need to spell out details in order to pursue our analysis.

By induction, we define the intended values \val{s} of (or denoted by) Dirac terms $s$ and check that the the following equivalences hold.

\begin{enumerate}[itemsep=2pt,leftmargin=25pt,parsep=0pt,topsep=5pt]
\item[E1]\ \val{s} is a linear function with domain \H\ $\iff$ $s$ ends with a bra, and
\item[]\ \val{s} is a linear function with domain \C\ $\iff$ $s$ ends with a function ket.
\item[E2]\ \val{s} is a vector or a scalar $\iff$ $s$ ends with a vector ket.
\item[E3]\ \val{s} is a scalar or a scalar-valued function $\iff$ $s$ starts with a bra.
\item[E4]\ \val{s} is a vector or a vector-valued function $\iff$ $s$ starts with a ket.
\end{enumerate}

The value \Val{\bra x} of a bra \bra x is the linear $\H\to\C$ function denoted by \bra x.
The ket values are described in the Tentative Convention above.
The equivalences E1--E4 are obvious in these cases.

Let $s$ be a concatenation $s_1 s_2$ of constituent subterms (which satisfy E1--E4 of course), and let $V_1, V_2$ be the values of $s_1, s_2$ respectively.
Four cases arise depending on whether $V_1, V_2$ are functions or not.

\begin{enumerate}[itemsep=0pt,leftmargin=30pt,parsep=0pt,topsep=5pt]
\item[FF]\ If $V_1$ is a function and $V_2$ is a function, then \val{s} is the composition $V_1 \circ V_2$, so that \val{s} is a function whose domain is that of $V_2$ and
    \[
    \val{s}(A) = \big(V_1\circ V_2\big)(A) = V_1(V_2(A)) \qquad
    \text{for all }A\in\Dom{V_2}.
   \]

\item[FN]\ If $V_1$ is a function but $V_2$ is not then
    \[ \val{s} = V_1(V_2). \]

\item[NF]\ If $V_1$ isn't a function but $V_2$ is then \val{s} is the function $V_1\cdot V_2$ so that \val{s} is a function whose domain is that of $V_2$ and
    \[\val{s}(A) = \big(V_1 \cdot V_2\big)(A) = V_1 \cdot (V_2(A))\qquad
    \text{for all }A\in\Dom{V_2}. \]

\item[NN]\ If neither $V_1$ nor $V_2$ is a function then
    \[ \val{s} = V_1\cdot V_2. \]
\end{enumerate}

\Qn I see that you overload the multiplication symbol $\cdot$ with different types.

\Ar We do. But notice that at least one of the factors is always a scalar.
It is very common to multiply scalars, vectors, and linear functions by scalars.
But let's check that our definitions make sense and that the equivalences E1--E4 hold.

\medskip
Clauses FF and FN make sense in that $V_2$ belongs to or takes values in \Dom{V_1}.
Indeed, by E1, $s_1$ ends with a bra or a function ket.
If $s_1$ ends with a bra, then $\Dom{s_1} = \H$ by E1 for $s_1$, and $s_2$ starts with a ket, and the desired property follows from E4 for $s_2$.
If $s_1$ ends with a function ket, then $\Dom{s_1} = \C$ by E1 for $s_1$, and $s_2$ starts with a bra, and the desired property follows from E3 for $s_2$.

Clauses NF and NN also make sense, which is obvious if $V_1$ is a scalar.
Otherwise $V_1$ is a vector and it suffices to check that $V_2$ is a scalar or scalar-valued function.
By E2, $s_1$ ends with a ket.
Hence $s_2$ starts with a bra.
Use E3 for $s_2$.

It remains to check that \val{s} satisfies the equivalences E1--E4.
By FF--NN, \val{s} is a function if and only if $V_2$ is a function, and if \val{s} is a function then its domain is \Dom{V_2}, and if \val{s} is not a function then it is a vector or scalar.
And of course $s,s_2$ share the final character.
Hence E1 and E2 hold.

Further, by FF--NN, \val{s} is a vector or vector-valued function if and only if $V_1$ is so.
It follows that \val{s} is a scalar or scalar-valued function if and only if $V_1$ is so.
And of course $s,s_1$ share the first character.
Hence E3 and E4 hold.

\Qn Let me just note that your clause NF generalizes the definition of outer product \ketbra wv quoted in \S\ref{s:qn} provided that the vector spaces $V$ and $W$ coincide with \H.

\Sec{Associativity}

\Ar It turns out that parentheses are unnecessary in Dirac terms because the partial operation
\[ \val{s_1} * \val{s_2} =  \val{s_1 s_2} \]
on the values of Dirac terms is associative.
The operation is defined if the concatenation $s_1 s_2$ is a Dirac term, i.e.\ if  kets and bras alternate in $s_1 s_2$.

\Qn What if $\val{s_i} = \val{t_i}$?
Will we have that $\val{s_1 s_2} = \val{t_1 t_2}$?

\Ar Yes, because the concatenation clauses in the definition of
\val{s_1 s_2} are formulated in terms of \val{s_1} and \val{s_2},
without examining the terms $s_1$ and $s_2$.

\begin{lemma}\label{l:ass}
The partial operation $*$ is associative.
In other words, let $s_1, s_2, s_3$ be Dirac terms such that kets and bras alternate in the concatenation $s_1 s_2 s_3$ and let $V_1, V_2, V_3$ be the values of $s_1, s_2, s_3$ respectively. Then
\begin{equation}\label{ass}
(V_1 * V_2) * V_3 = V_1 * (V_2 * V_3)
\end{equation}
\end{lemma}

\begin{proof}
First examine $V_3$.
If $V_3$ is a scalar, factor it out of the equation, so that \eqref{ass} becomes obvious.
Similarly, if $V_3$ is a scalar-valued function, then
\begin{equation}\label{ass2}
(V_1 * V_2) * V_3(A)= V_1 * (V_2 * V_3(A)) \qquad
\text{holds for every }A\in\Dom{V_3}
\end{equation}
because the scalar $V_3(A)$ can be factored out of the equation.
If $V_3$ is a vector-valued function, it suffices to prove \eqref{ass2}
because it implies \eqref{ass}.
In order to prove \eqref{ass2} for vector-valued functions, it suffices to prove \eqref{ass} for the case where $V_3$ is a vector.
In this case, by E4, $s_3$ starts with ket, $s_2$ ends with a bra, and therefore, by E1, $V_2$ is a function with domain \H.

Next examine $V_1$.
If $V_1$ is a scalar, factor it out, and then \eqref{ass} is obvious.
If $V_1$ is a function then, using FF and FN in \S\ref{s:terms}, we have:
\begin{align*}
(V_1 * V_2) * V_3 &= \big(V_1 \circ V_2\big)(V_3) = V_1(V_2(V_3))\\
V_1 * (V_2 * V_3) &= V_1 * (V_2(V_3)) =  V_1(V_2(V_3))
\end{align*}
It remains to prove \eqref{ass} in the case where $V_1, V_3$ are vectors and $V_2$ is a function with domain \H.
Since $V_1$ is a vector, $s_1$ ends with a ket by E2, so that $s_2$ starts with a bra and $V_2$ is a scalar-valued function by E3.
We have
\begin{align*}
(V_1 * V_2) * V_3 &= \big(V_1\cdot V_2\big)(V_3)
 = V_1\cdot V_2(V_3),\\
V_1 * (V_2 * V_3) &= V_1 * (V_2(V_3)) = V_1\cdot V_2(V_3).
\qedhere
\end{align*}
\end{proof}

\Sec{Resolution}

\Qn What does the associativity buy you?

\Ar It allows us to prove a certain robustness phenomenon which can be illustrated on the example where
\[
  s = \big(\braket xy\big)\bra u
\]
Let \ket v be an arbitrary vector in \H\ and $c,d$ be the scalar products \braket xy and \braket uv, respectively.
If \ket y as
\blue{is}
a vector ket in $s$ then, by value, i.e., writing terms instead of their values (as it is commonly done)
\[
s\ket v = \Big(\big(\braket xy\big)\bra u\Big)\ket v =
  \big(c\cdot \bra u)\ket v =
  c\cdot d.
\]
If \ket y is a function ket in $s$ then $\bra x \circ \ket y$ is the the operation of multiplying by $c$, and so (again by value) we have
\[
s\ket v = \Big(\bra x\, \ket y\, \bra u\Big)\ket v =
\Big(\bra x \circ \ket y \circ \bra u\Big)\ket v =
\Big(\bra x \circ \ket y\Big) \cdot d = c\cdot d,
\]
getting exactly the same result.

\begin{lemma}[Robustness]
Let $s$ be a Dirac term and let a ket \ket y occur in a particular non-final position in sequence $s$.
\val{s} is the same whether (the occurrence of) \ket y in that position is a vector ket or a function ket.
\end{lemma}

\begin{proof}
First suppose that \ket y is the first character in $s$.
By Lemma~\ref{l:ass}, we may assume that $s$ is a concatenation of \ket y and some Dirac term $s_2$. Let $V, V_2$ be the values of \ket y and $s_2$ respectively.
By E3, $V_2$ is a scalar or a scalar-valued function.
Recall that there is a vector \v\ such that $y$ is a marked version of \v.
If \ket y is a vector ket then $V = \v$, and if \ket y is a function ket it is the function $V(a) = a \cdot \v$.

If $V_2$ is a scalar then, by NF--NN in \S\ref{s:terms},
\[ \val{s} =
\begin{cases}
  V(V_2) = V_2\cdot\v\
    &\text{if \ket y is a function ket,}\\
  V\cdot V_2 = \v\cdot V_2 = V_2 \cdot \v\
    &\text{if \ket y is a vector ket.}
\end{cases}
\]
Similarly, if $V_2$ is a scalar-valued function then, by FF and FN, for every argument $A$ of $V_2$, we have
\[ \val{s}(A) =
\begin{cases}
  V(V_2(A)) = V_2(A)\cdot\v\
    &\text{if \ket y is a function ket,}\\
  V\cdot V_2(A) = \v\cdot V_2(A) = V_2(A) \cdot \v\
    &\text{if \ket y is a vector ket.}
\end{cases}
\]
This completes the proof in the case that \ket y is at the beginning
of $s$.

Now suppose that $s = s_1 \ket y s_2$ where $s_1, s_2$ are Dirac terms with values $V_1, V_2$.
By Lemma~\ref{ass}, we may assume that $s$ is the concatenation of $s_1$ and $\ket y s_2$, so that \val{s} is determined by $V_1$ and \Val{\ket y s_2}.
By the first part of the proof, \Val{\ket y s_2} does not depend on
how the ket \ket y is marked. It follows that \val{s} does not depend on how \ket y is marked.
\end{proof}

Now we drop the Tentative Convention of \S\ref{s:terms}.
The kets are not marked anymore.
One should be able to tell from the context whether a ket denotes a vector or a function.

\Qn By the robustness lemma, we have a whole spectrum of possible resolutions of the abuse of notation in question.

\Ar One natural resolution is to view kets as function kets where possible:

\textit{In a Dirac term, an occurrence of a ket is viewed as a vector if and only if it is the final character in the term.}

\Qn The direct opposite strategy is to view an occurrence of a ket as a function if and only if it is the first character in the term.
I'm kidding.

\Ar Actually, a close relative of your strategy works: View an occurrence of a ket as a function if and only if it is the first character and the last character is a bra.

\Qn Explain.

\Ar If the first character is a bra or if the last character is a ket, then the given Dirac term has the form
\[ \bra{x_1}\,\ket{y_1}\,\dots\,\bra{x_n}\,\ket{y_n},\quad
\bra{x_1}\,\ket{y_1}\,\dots\,\bra{x_n}\,\ket{y_n}\,\bra{x_0},\quad
\text{ or }\quad \ket{y_0}\,\bra{x_1}\,\ket{y_1}\,\dots\,\bra{x_n}\,\ket{y_n}. \]
Pair up every \bra{x_i} with \ket{y_i} and let $c = \prod_i \braket{x_i}{y_i}$.
Every ket is viewed as a vector, and you get $c$, $c\cdot\bra{x_0}$, or $c\cdot\ket{y_0}$ respectively.

If the first character is a ket and the last character is a bra, the term has the form
\[\ket{y_0}\,\bra{x_1}\,\ket{y_1}\,\dots\,\bra{x_n}\,\ket{y_n}\,%
\bra{x_0} . \]
Pair up \bra{x_1}, \dots, \bra{x_n} with \ket{y_1}, \dots, \ket{y_n} respectively and let $c = \prod_{i=1}^n \braket{x_i}{y_i}$.
You get $c\cdot\ketbra{y_0}{x_0} = c\cdot \ket{y_0} \circ \bra{x_0}$.

\end{document}